\documentclass[11pt, reqno, letterpaper]{amsart}

\usepackage{graphicx}    
\usepackage{bbm, xcolor}
\usepackage{xfrac}
\usepackage{multicol}
\usepackage{soul}

\usepackage[margin=1in]{geometry}

\usepackage[english]{babel}

\usepackage{amsmath,amsthm,amsfonts,amssymb}

\DeclareMathOperator\arctanh{arctanh}

\newtheorem{theorem}{Theorem}[section]

\newtheorem{lemma}{Lemma}[section]
\newtheorem{proposition}{Proposition}[section]
\newtheorem{assumption}{Assumption}[section]

\newtheorem{remark}{Remark}[section]
\numberwithin{figure}{section}

\begin{document}

\title[VIX options in the SABR model]
{VIX options in the SABR model}

\author{Dan Pirjol}
\address{Stevens Institute of Technology, Hoboken, New Jersey, United States of America}

\author{Lingjiong Zhu}
\address{Florida State University, Tallahassee, Florida, United States of America}

\date{July 15, 2025}


\begin{abstract}
We study the pricing of VIX options in the SABR model $dS_t = \sigma_t S_t^\beta dB_t, d\sigma_t = \omega \sigma_t dZ_t$ where $B_t,Z_t$ are standard Brownian motions correlated with correlation $\rho<0$ and
 $0 \leq \beta < 1$. VIX is expressed as a risk-neutral conditional expectation of an integral over the volatility process $v_t = S_t^{\beta-1} \sigma_t$. We show that $v_t$ is  the unique solution to a one-dimensional diffusion process. Using the Feller test, we show that $v_t$ explodes in finite time with non-zero probability. As a consequence, VIX futures and VIX call prices are infinite, and VIX put prices are zero for any maturity. As a remedy, we propose a capped volatility process by capping the drift and diffusion terms in the $v_{t}$ process such that it becomes non-explosive and well-behaved, 
and study the short-maturity asymptotics for the pricing of VIX options.
\end{abstract}

\maketitle

\baselineskip18pt

\section{Introduction}

The CBOE Volatility Index (VIX) is a measure of the S\&P 500 expected volatility, and is computed and published by the 
Chicago Board Options Exchange (CBOE). This index is defined by the risk-neutral expectation
\begin{equation}
\mathrm{VIX}_T^2 = \mathbb{E}\left[-\frac{2}{\tau} \log\left(\frac{S_{T+\tau}}{S_T}\right)+\frac{2}{\tau}\int_{T}^{T+\tau}\frac{dS_{t}}{S_{t}} \Big|\mathcal{F}_T \right],
\end{equation}
where $S_t$ is the equity index S\&P 500 at time $t$, and $\tau = 30$ days. This expectation is estimated from the prices of current (as of $T$)
call and put options on the SPX index, and
is estimated from the traded prices of options on the S\&P 500 index.
The methodology for VIX computation is detailed in the VIX White Paper \cite{VIXwp}.

CBOE lists futures and options on the $\mathrm{VIX}_T$ index with
several maturities.
VIX option contracts are cash settled with an amount linked to the $\mathrm{VIX}_T$ observed at the contract maturity $T$.
Denote $r,q$ the risk-free rate and the dividend yield, respectively.
Under a model for $S_{t}$ 
with continuous paths (no jumps) of the form
\begin{equation}\label{Sgeneric}
\frac{dS_t}{S_t} = v_t dB_t + (r-q) dt\,, 
\end{equation}
with $\{v_t\}_{t\geq 0}$ being the \emph{volatility process} that follows a stochastic process, $B_{t}$ a standard Brownian motion, 
the $\mathrm{VIX}_T$ index is given by the risk-neutral expectation
\begin{equation}\label{VIXdef}
\mathrm{VIX}_T^2 = \mathbb{E}\left[\frac{1}{\tau} \int_T^{T+\tau} v_t^2 dt \Big|\mathcal{F}_T\right].
\end{equation}

The price of a futures contract on the $\mathrm{VIX}$ index with maturity $T$ is given by the risk-neutral expectation
\begin{equation}\label{VIX:F}
F_V(T) = \mathbb{E}[\mathrm{VIX}_T] \,,
\end{equation}
and the prices of VIX calls and puts with strike $K$ and maturity $T$ are given by 
\begin{align}\label{VIX:call}
C_{V}(K,T)= e^{-rT } \mathbb{E}[(\mathrm{VIX}_{T}-K)^{+}]\,,\quad
P_{V}(K,T)= e^{-rT} \mathbb{E}[(K-\mathrm{VIX}_{T})^{+}] \,.
\end{align}

The VIX option pricing has been extensively studied in the literature.
Ref.~\cite{Detemple2000} studied volatility options pricing under several popular diffusion models for the variance process. 
\cite{Carr2005} derived results for volatility options under pure jump models with independent increments.
\cite{Sepp2008a,Sepp2008} studied volatility derivatives under a square root volatility model with jumps. 
\cite{Goard2013} derived
analytical results for VIX options in the 3/2 stochastic volatility model, 
which was extended to allow jumps in \cite{Baldeaux2014}. 
Recently, the martingale optimal transport approach was applied to the problem of the joint SPX/VIX smile calibration in \cite{Guyon2020,Guyon2024}. 
We mention also additional recent work on VIX option pricing presented in 
\cite{AbiJaber2022,Alos2022VIX,Forde2022,Cao2020,Tong2021,Yuan2022,Cuchiero2023}.

VIX option pricing was studied under local-stochastic volatility (LSV) models in \cite{Forde2023,VIXpaper}. This includes the log-normal ($\beta=1$) SABR model as a special case. 
Recall that the SABR model \cite{SABRpaper} is defined by
\begin{equation}\label{SABR}
dS_t  =  S_t^\beta \sigma_t dB_t + (r-q) S_t dt\,,\qquad
d\sigma_t = \omega \sigma_t dZ_t\,,
\end{equation}
where $B_t,Z_t$ are correlated standard Brownian motions 
with correlation $\rho\in[-1,1]$,  
$\omega>0$ is the volatility of volatility parameter, 
and $\beta\in[0,1]$ is an exponent which controls
the backbone of the implied volatility \cite{SABRpaper}.
For simplicity we take $r=q=0$.

In this paper we study the pricing of VIX options in the SABR model with $0\leq\beta<1$. 
The main motivation for this study is theoretical - the SABR model is a particular case of the LSV model studied in \cite{VIXpaper} although it does not satisfy the technical conditions of that paper. Thus it is of interest to see if the results of \cite{VIXpaper} extend also to this case. This turns out to be indeed the case, although we show that it exhibits also some surprisingly different  qualitative behaviors.

Another motivation is the widespread use of the model in financial practice. For this reason, many results are available: Hagan et al. \cite{SABRpaper} derived short maturity asymptotics for the asset price distribution, option prices and implied volatility at leading order in the SABR model. 
The expansion was extended to higher order in \cite{HLbook,Lewis1} and the convergence of the expansion was studied in \cite{Lewis2022}.
The martingale properties of the $\beta=1$ SABR model were studied in \cite{Lions2008}.
The simulation and option
pricing under the SABR model have also been studied extensively, see \cite{AKSbook,HLbook} for reviews.
An exact simulation method for the SABR model for $\beta=1$ and for $\beta<1$, $\rho=0$ was proposed in \cite{CSC2017}.
An efficient simulation method was proposed and studied in \cite{CKN2021} for SABR and more general stochastic local volatility models based on continuous-time Markov chain approximation.

The SABR model \eqref{SABR} is of the type (\ref{Sgeneric}) with volatility process 
$v_t=S_t^{\beta-1} \sigma_t$.
The SABR model \eqref{SABR} is also a special case of the local-stochastic volatility model studied in \cite{VIXpaper}
\begin{align}
\frac{dS_t}{S_t} = \eta(S_t) \sqrt{V_t} dB_t\,, \quad
\frac{dV_t}{V_t} = \mu(V_t) dt + \sigma(V_t) dZ_t\,,
\end{align}
where $(B_t,Z_t)$ are correlated standard Brownian motions, with correlation $\rho$. The SABR model (\ref{SABR}) is recovered by taking $V_t=\sigma_t^2$ and $\eta(x) = x^{\beta-1}$, $\mu(v)=\omega^2$, $\sigma(v)=2\omega$.
However, for $0\leq \beta < 1$, this choice of $\eta(\cdot)$ does not satisfy the technical assumptions
in \cite{VIXpaper} and hence the results of that paper are not directly applicable.


We will show in this paper that the SABR model with $0\leq\beta<1$ exhibits
qualitatively different behavior that leads
to explosion for $\mathrm{VIX}_{T}$. 
We propose a modification of the model based on capping the volatility process which ensures well-behaved VIX futures and option prices. 


\section{The Volatility Process in the SABR Model}

Consider the process for $\mathrm{VIX}_T$, defined as in (\ref{VIXdef}). 
In the SABR model \eqref{SABR}, we have
\begin{equation}\label{VIX:v:eqn}
\mathrm{VIX}_{T}^{2}=\mathbb{E}\left[\frac{1}{\tau}\int_{T}^{T+\tau}v_{t}^{2}dt\Big|\mathcal{F}_{T}\right],   
\end{equation}
where $v_{t}$ is the \emph{volatility process} defined as
\begin{equation}
v_{t}:=S_{t}^{\beta-1}\sigma_{t}, 
\end{equation}
where $S_{t},\sigma_{t}$ are given in \eqref{SABR}.
An application of It\^{o}'s lemma gives that $v_t$ follows 
a one-dimensional stochastic differential equation (SDE).
\footnote{This observation was also noted in \cite{Lewis2}, see Section 8.3.2 where $v_t$ is called the effective log-normal volatility process for the SABR model.}

\begin{lemma}\label{lem:existence}
In the SABR model \eqref{SABR}, $v_t=\sigma_t S_t^{\beta-1}$ follows the SDE:
\begin{eqnarray}\label{vsde}
&&\frac{dv_t}{v_t} =  \sigma_V(v_t) d W_t + \mu_V(v_t) dt\,,
\end{eqnarray}
where $W_t$ is a standard Brownian motion different from $B_t$,
and the volatility and drift of this process are given by 
\begin{eqnarray}
&& \sigma_V(v) := \sqrt{\omega^2 + (\beta-1)^2 v^2 + 2\rho (\beta-1) \omega v}\,,\label{sigma:v:eqn} \\
&& \mu_V(v) := v (\beta-1)\left[ \frac12 (\beta - 2) v + \rho \omega \right] \,.\label{mu:v:eqn}
\end{eqnarray}
\end{lemma}

\begin{proof}
We have by It\^{o}'s lemma
\begin{align}
dv_t &= d(\sigma_t S_t^{\beta-1} ) \\
&= S_t^{\beta-1} d\sigma_t + 
\sigma_t dS_t^{\beta-1} +  \rho \frac{\partial^2 (\sigma_t S_t^{\beta-1})}{\partial\sigma_t \partial S_t} dt\nonumber\\
&= S_t^{\beta-1} \omega \sigma_t dZ_t + \sigma_t \Big( (\beta-1) \sigma_t S_t^{2\beta-2} dB_t + \frac12 (\beta-1)(\beta-2) \sigma_t^2 S_t^{3\beta-3} dt\Big) \nonumber \\
&\qquad\qquad\qquad+  \rho \omega (\beta-1) \sigma_t^2 S_t^{2\beta-2} dt\,, \nonumber
\end{align}
where the last term was computed as
\begin{align}
\frac{\partial^2 (\sigma_t S_t^{\beta-1})}{\partial\sigma_t \partial S_t} =(\beta-1) S_t^{\beta-1} (\omega \sigma_t) (\sigma_t S_t^\beta) =
\omega (\beta-1) \sigma_t^2 S_t^{2\beta-2}\,.
\end{align}
This can be written equivalently as
\begin{align}
    \frac{dv_t}{v_t} &= \omega dZ_t + (\beta-1) v_t dB_t + 
    (\beta-1) v_t \Big( \rho \omega + \frac12 (\beta-2) v_t\Big) dt \\
    &= \sigma_V(v_t) dW_t + \mu_V(v_t) dt\,, \nonumber
\end{align}
where $W_t:=\int_{0}^{t}\frac{ \omega dZ_s + (\beta-1) v_s dB_s}{\sigma_V(v_{s})}$ is another standard Brownian motion due to L\'{e}vy's characterization of Brownian motion.
This reproduces the stated result (\ref{vsde}).
\end{proof}

We will show in Section~\ref{sec:existence:uniqueness} that under the assumption that $0\leq\beta<1$ and $\rho<0$, the SDE \eqref{vsde} has a unique solution.
In Section~\ref{sec:2.4} we prove that $(v_{t})_{t\geq 0}$
is explosive, i.e. $\mathbb{P}(\tau_{\infty}<\infty)>0$
where $\tau_{\infty}:=\sup\{t\geq 0:v_{t}<\infty\}$. 
As an immediate consequence, for any $\tau>0$, by the definition of $\mathrm{VIX}_{T}$ in \eqref{VIX:v:eqn}, we have $\mathrm{VIX}_{T}=\infty$
with probability one.
As a result, the VIX forward price $F_V(T)$ and the
VIX call option prices are infinite, and the VIX put option prices vanish.
In Section~\ref{sec:3} we will present a modification of the process \eqref{vsde} which avoids these issues.

\subsection{Existence and Uniqueness of the Volatility Process}\label{sec:existence:uniqueness}

We study in this section the existence and uniqueness for the volatility process $v_{t}$ in (\ref{vsde}). 
Note that since the solution to the SABR model 
exists\footnote{Conditional on a realization of the volatility process $\sigma_t$, the SABR model with $0\leq \beta < 1$ reduces to the CEV model, which is related to a squared Bessel process by a change of variable \cite{Linetsky2010}. The conditions for existence and uniqueness of the solutions of this model are given for example in Section 2 in Chen, Oosterlee, van der Weide (2012) \cite{Chen2012}; see also \cite{Andersen2000}.}, 
Lemma~\ref{lem:existence} already shows the existence
of the solution for the SDE \eqref{vsde}.
This SDE has the form
\begin{equation}
dv_t = (\alpha v_t^3 + \delta v_t^2) dt + v_t \sigma_V(v_t) dW_t\,,
\end{equation}
with coefficients 
\begin{equation}\label{alpha:delta}
\alpha := \frac12 (1-\beta)(2-\beta)\,,\quad
\delta := -(1-\beta) \rho \omega \,.
\end{equation}

In the rest of this section, we make the following assumption on the model parameters.

\begin{assumption}\label{assump:1}
Assume $0 \leq \beta < 1$ and $-1<\rho<0$. 
\end{assumption}
With this assumption, the coefficients in the drift \eqref{alpha:delta} are positive, i.e. $\alpha,\delta >0$ and $\sigma_V(v)>0$ for all real $v$.
The drift and volatility in \eqref{sigma:v:eqn}-\eqref{mu:v:eqn} are not sub-linear so we cannot use Theorem~2.9 in Karatzas and Shreve \cite{KS} to prove the existence and uniqueness.
Instead, we will use Theorem~5.15 in Karatzas and Shreve \cite{KS} to show the existence and uniqueness in the weak sense.
Consider the SDE 
\begin{equation}\label{general:SDE}
dX_t = b(X_t) dt + \sigma(X_t)dW_t, 
\end{equation}
satisfying non-degeneracy \eqref{ND} and local integrability \eqref{LI} conditions,
defined as
\begin{equation}\label{ND}\tag{ND}
\sigma^2(x) > 0\,,\qquad\text{for every $x \in (0,\infty)$}\,,
\end{equation}
and
\begin{equation}\label{LI}\tag{LI}
\text{for any $x \in (0,\infty)$, there exists some $\varepsilon>0$ such that $\int_{x-\varepsilon}^{x+\varepsilon} \frac{|b(y)| dy}{\sigma^2(y)} < \infty$}\,.
\end{equation}

For our case we have $b(x) = \alpha x^3 + \delta x^2$ and $\sigma(x) = x \sigma_V(x)$, where $\alpha,\delta$ are defined in \eqref{alpha:delta} and $\sigma_{V}(\cdot)$ is defined in \eqref{sigma:v:eqn}.
These conditions are satisfied by our process
under Assumption~\ref{assump:1}, which ensures the non-vanishing of $\sigma_V(x)$. 

By Theorem~5.15 in Karatzas and Shreve \cite{KS} (reproduced below for convenience) the solution of \eqref{vsde} exists and is unique in weak sense.

\begin{theorem}[Theorem~5.15 in Karatzas and Shreve \cite{KS}]
Assume that $\sigma^{-2}$ is locally integrable at every point in $[0,\infty)$, and conditions~\eqref{ND} and \eqref{LI} hold. 
Then for every initial distribution for $X_0$, the SDE~\eqref{general:SDE} for $X_t$ has a weak solution up to an explosion time, 
and this solution is unique in the sense of probability law.
\end{theorem}

\subsection{Reduction to the Natural Scale}
\label{sec:2.2}

We present in this section the reduction of the diffusion \eqref{vsde} to its natural scale. The approach used is described, for example, in Section 5.5.B (p.339) of Karatzas and Shreve \cite{KS}. We introduce and study several functions which will be required for the application of the Feller explosion criterion in Section \ref{sec:2.4}.

Define the \textit{scale} function
\begin{equation}\label{pdef}
p(x) := \int_c^x e^{-2 \int_c^\xi \frac{b(y) dy}{\sigma^2(y)} } d\xi\,,
\end{equation}
with $c$ an arbitrary value in $[0,\infty)$. 
By Proposition~5.13 in \cite{KS}, $Y_t = p(X_t)$ follows the driftless process $dY_t = \bar \sigma(Y_t) dW_t$ with volatility:
\begin{equation}
\bar\sigma(y) :=
\begin{cases}
p'(q(y)) \sigma(q(y)) & \text{if $y \in (0,p_\infty)$}, \\
0 & \text{otherwise}, \\
\end{cases}
\end{equation}
where $q(y)$ is the inverse of $p(x)$ and $p_\infty=\lim_{x\rightarrow\infty}p(x)$ (see Proposition~\ref{prop:p}). 
The diffusion process $Y_t$ is in its natural scale, since its scale function is simply $x$.

Denote the integral in the exponent of (\ref{pdef}) as
\begin{equation}
F(x;c) := \int_c^x \frac{b(y) dy}{\sigma^2(y)} = 
\int_c^x \frac{\alpha y + \delta}{\sigma_V^2(y)} dy \,.
\end{equation}
The denominator does not vanish for all real $y$.
Thus the integrand is well-behaved as $y\to 0$ and we can take without any loss of generality $c=0$.
The value of the integral for any other $c>0$ can be recovered as $F(x;c) = F(x;0) - F(c;0)$.
Denote for simplicity $F(x) := F(x;0)$, which can be evaluated in closed form with the result 
\begin{eqnarray}
F(x) =\frac{1}{(1-\beta)^2} \left\{
\frac{\alpha}{2} \log \frac{R(x)}{\omega^2} 
- \frac12 \beta(1-\beta) \frac{|\rho|}{\rho_\perp}
\left( \arctan\frac{(1-\beta) x + \bar\omega}{\omega\rho_\perp} - \arctan \frac{\bar\omega}{\omega\rho_\perp}\right) \right\}\,,
\end{eqnarray}
where 
\begin{equation}\label{R:eqn}
R(x) := \sigma_V^2(x) = \omega^2  + 2 \bar\omega (1-\beta) x + (1-\beta)^2 x^2\,,
\end{equation}
with
$\bar\omega := \omega |\rho|$ and $\rho_\perp := \sqrt{1-\rho^2}$.
We note the lower and upper bounds on $e^{-2F(x)}$:
\begin{equation}\label{boundsF}
\left( \frac{\omega^2}{R(x)} \right)^{\frac{\alpha}{(1-\beta)^2}} \leq
e^{-2F(x)} \leq
\kappa
\left( \frac{\omega^2}{R(x)} \right)^{\frac{\alpha}{(1-\beta)^2}}\,,
\end{equation}
where 
\begin{equation}\label{kappa:defn}
\kappa := \exp \left( \frac{\pi}{2} \frac{\beta}{1-\beta} \frac{|\rho|}{\rho_\perp}\right) > 1\,.
\end{equation}

\begin{proposition}\label{prop:p}
Take $c=0$ and assume the conditions in Assumption~\ref{assump:1} for $\beta,\rho$.
The scale function $p(x)$ has the following properties:

i) $p(0)=0$. $p(x)$ is monotonically increasing and approaches a finite limit 
$\lim_{x\to \infty} p(x) = p_\infty < \infty$ as $x\to \infty$.

ii) The large $x$ asymptotics has the form
\begin{equation}\label{largex}
p(x)=p_\infty - \frac{c_1}{x^{\frac{1}{1-\beta}}}+o\left(x^{-1/(1-\beta)}\right)\,,
\end{equation}
as $x\rightarrow\infty$, with $c_1>0$ a positive constant.

iii) The inverse function $q(y)$ diverges to $+\infty$ as
\begin{equation*}
q(y) \sim \left(\frac{c_1}{p_\infty - y} \right)^{1-\beta}\,,
\end{equation*}
as $y\to p_\infty$.
\end{proposition}
\begin{proof}
    The proof is given in the Appendix. 
\end{proof}


\begin{remark}
The choice $c=0$ in the definition of the scale function is not essential and can be relaxed. Denoting $p(x;c)$ the scale function with an arbitrary value of $c\geq 0$, this is related to $p(x)$ defined with $c=0$ as
\begin{equation}
p(x;c) = e^{-2F(c)} (p(x) - p(c)) \,.
\end{equation}
\end{remark}

\begin{remark}\label{prop:sig}
Using the asymptotics of $p(x)$ and $q(y)$ obtained in Proposition \ref{prop:p} we obtain the following properties of the volatility in natural scale $\bar\sigma(y)$:

i) For small $y\to 0$ argument it has the form $\bar\sigma(y) = \omega y + O(y^2)$.

ii) As $y\to p_\infty$ it has the asymptotics $\bar\sigma(y) \sim (p_\infty - y)^\beta$.
\end{remark}



The diffusion in natural scale $Y_t$ is bounded between $0$ and $p_\infty$ and is non-explosive. 
Since the map $q(y) \to \infty$ as $y\to p_\infty$, the process $v_t$ in \eqref{vsde} explodes at the first hitting time of $Y_t$ to level $p_\infty$.


\subsection{Nature of Boundary Points}

We established above that the process for $Y_t = p(X_t)$
takes values in the bounded range $Y_t \in [0,p_\infty]$. The asymptotics of the volatility in the natural scale $\bar \sigma(y)$ from Proposition~\ref{prop:sig} determines the nature of the boundary points.

The $y=0$ boundary is similar to that of a geometric Brownian motion process and is a natural boundary.

The point $y=p_\infty$ is similar to the origin for the CEV process. Recall the classification of the solutions for this case \cite{Linetsky2010}: 

a) $\beta \in (0,\frac12)$. 
The point $y=p_\infty$ is a regular boundary point. The fundamental solution is not unique.
The problem is well-posed only if an additional boundary condition is imposed, for example, absorbing or reflecting boundary.

b) $\beta \in [ \frac12, 1)$. 
The point $y=p_\infty$ is an exit boundary point.
There is a unique fundamental solution with decreasing norm and mass at $y=p_\infty$
corresponding to the absorption at this point.

\subsection{Feller Test of Explosion for the Volatility Process}
\label{sec:2.4}

In this section, we study the explosion for the volatility process $v_{t}$ in \eqref{vsde}.
A sufficient condition for the absence of explosions is expressed in terms of the scale function as $p_{\infty} = \infty$ \cite{KS}. 
By point (i) in Proposition \ref{prop:p}, the $x\to \infty$ limit of the scale function is finite, so explosions cannot be excluded. However the finiteness of $p_\infty$ is only a necessary, but not sufficient criterion for the existence of explosions. 
Using the Feller test of explosion \cite{Feller1952}, 
we show in this section that the process $v_t$ in \eqref{vsde} explodes in finite time with non-zero probability. 
A heuristic proof of this result is also given in Section 8.5 of \cite{Lewis2}.

\begin{proposition}\label{prop:explosion}
The process $v_t$ in \eqref{vsde} explodes in finite time with non-zero probability.
\end{proposition}
\begin{proof}
    The proof is given in the Appendix. 
\end{proof}

\begin{remark}\label{remark:explosion}
Since the process for $v_t$ in \eqref{vsde} explodes in finite time with non-zero probability, we have
$\mathbb{E}[(v_t)^p]=\infty$ for all $p>0$. In particular, the expectation $\mathbb{E}[v_t]$ is infinite.
\end{remark}

We give also a simple proof that $S_t$ is a true martingale in the SABR model for all $\beta<1$. This was proved heuristically in Section 8.5 of \cite{Lewis2}. The proof is given in the Appendix.
\begin{proposition}\label{prop:martingale}
Assume $\beta<1$. Then the asset price $S_t$ in the SABR model is a true martingale such that $\mathbb{E}[S_t] = S_0$ for all $t>0$.
\end{proposition}

\section{Capped Volatility Process and VIX Option Pricing}
\label{sec:3}

The result of Proposition~\ref{prop:explosion} is a surprising negative result. It implies that the VIX futures prices $F_V(T)$ in \eqref{VIX:F} and VIX call option prices $C_V(K,T)$ in \eqref{VIX:call} are infinite. On the other hand the VIX put option prices $P_V(K,T)$ are zero for any maturity $T>0$ and strike $K>0$. This limits the practical usefulness of the SABR model with $\beta<1$ for pricing these products. 

On the other hand, the SABR model with $\beta=1$ has well behaved VIX options and futures prices. The predictions for this case have been discussed in \cite{Forde2023} and \cite{VIXpaper}.

As a remedy for the SABR model with $\beta<1$, we propose to cap
the drift and diffusion terms in the $v_{t}$ process to make it non-explosive
such that one can price the VIX options in practice.
In particular, based on the volatility process $v_{t}$ in \eqref{vsde}, 
we propose the following modification, the \emph{capped volatility process}:
\begin{eqnarray}\label{vsde:modification}
&&\frac{dv_t}{v_t} =  \hat{\sigma}_V(v_t) d W_t + \hat{\mu}_V(v_t) dt\,,
\end{eqnarray}
with \emph{capped} volatility and drift
\begin{eqnarray}
&& \hat{\sigma}_V(v) := \min\left(a,\sigma_V(v)\right)\,, 
\\
&& \hat{\mu}_V(v) := \min\left(b,\mu_V(v)\right)\cdot 1_{\mu_V(v)>0}
+\max\left(-b,\mu_V(v)\right)\cdot 1_{\mu_V(v)\leq 0}\,,
\end{eqnarray}
where $\sigma_V(v)$ and $\mu_V(v)$ are given in \eqref{sigma:v:eqn}-\eqref{mu:v:eqn} and $a,b>0$ are given caps.
Under Assumption~\ref{assump:1}, $\sigma_{V}(v)$ is monotonically
increasing in $v\geq 0$, and we assume that $a>\sigma_{V}(0)=\omega$
such that
\begin{equation}
\hat{\sigma}_V(v)=
\begin{cases}
\sigma_{V}(v)=\sqrt{\omega^2 + (\beta-1)^2 v^2 + 2\rho (\beta-1) \omega v} &\text{if $v\leq\hat{v}$},
\\
a &\text{if $v>\hat{v}$},
\end{cases}
\end{equation}
where
\begin{equation}\label{vhat}
\hat{v}:=\frac{\rho\omega+\sqrt{a^{2}+(\rho^{2}-1)\omega^{2}}}{1-\beta}.
\end{equation}

With $v_{t}$ defined in \eqref{vsde:modification},
the $T,\tau\to 0$ asymptotics of the VIX call and put options can be achieved by quoting the results 
of the short-maturity European call and put options for the local volatility model. 
First, we will show that $\mathrm{VIX}_{T}\rightarrow v_{T}$ almost surely
as $\tau\rightarrow 0$, and indeed we have the following result.

\begin{proposition}\label{prop:VIX:approx}
For any $\tau,T>0$, we have
\begin{equation}
v_{T}e^{-b\tau}\leq\mathrm{VIX}_{T}\leq v_{T}e^{b\tau+\frac{1}{2}a^{2}\tau}.
\end{equation}
\end{proposition}

\begin{proof}
First, we notice that for any $v$, 
$0\leq\hat{\sigma}_{V}(v)\leq a$
and $|\hat{\mu}_{V}(v)|\leq b$.
Next, it follows from \eqref{vsde:modification} that
for any $t\geq T$, 
\begin{equation}
v_{t}=v_{T}e^{\int_{T}^{t}(\hat{\mu}_{V}(v_{s})-\frac{1}{2}\hat{\sigma}_{V}^{2}(v_{s}))ds+\int_{T}^{t}\hat{\sigma}_{V}(v_{s})dW_{s}}.
\end{equation}
Therefore, we have
\begin{align*}
\mathbb{E}[v_{t}^{2}|\mathcal{F}_{T}]
&=v_{T}^{2}\mathbb{E}\left[e^{\int_{T}^{t}(2\hat{\mu}_{V}(v_{s})-\hat{\sigma}_{V}^{2}(v_{s}))ds+2\int_{T}^{t}\hat{\sigma}_{V}(v_{s})dW_{s}}|\mathcal{F}_{T}\right]
\\
&\leq
v_{T}^{2}e^{2b(t-T)+a^{2}(t-T)}\mathbb{E}\left[e^{-\int_{T}^{t}\frac{1}{2}(2\hat{\sigma}_{V})^{2}(v_{s})ds+\int_{T}^{t}2\hat{\sigma}_{V}(v_{s})dW_{s}}|\mathcal{F}_{T}\right]
=v_{T}^{2}e^{2b(t-T)+a^{2}(t-T)},
\end{align*}
and similarly, 
\begin{align*}
\mathbb{E}[v_{t}^{2}|\mathcal{F}_{T}]
\geq
v_{T}^{2}e^{-2b(t-T)}\mathbb{E}\left[e^{-\int_{T}^{t}\frac{1}{2}(2\hat{\sigma}_{V})^{2}(v_{s})ds+\int_{T}^{t}2\hat{\sigma}_{V}(v_{s})dW_{s}}|\mathcal{F}_{T}\right]
=v_{T}^{2}e^{-2b(t-T)}.
\end{align*}
Hence, we conclude that
\begin{align*}
\mathrm{VIX}_{T}
\leq\left(v_{T}^{2}\frac{1}{\tau}\int_{T}^{T+\tau}e^{2b(t-T)+a^{2}(t-T)}dt\right)^{1/2}
\leq v_{T}e^{b\tau+\frac{1}{2}a^{2}\tau},
\end{align*}
and
\begin{align*}
\mathrm{VIX}_{T}
\geq\left(v_{T}^{2}\frac{1}{\tau}\int_{T}^{T+\tau}e^{-2b(t-T)}dt\right)^{1/2}
\geq v_{T}e^{-b\tau}.
\end{align*}
This completes the proof.
\end{proof}

Proposition~\ref{prop:VIX:approx} implies that $|\mathrm{VIX}_{T}-v_{T}|=O(\tau)$ almost surely as $\tau\rightarrow 0$.
We can leverage this result and the short-maturity asymptotics
for European options in the local volatility model \cite{BBFpaper} to obtain the short-maturity asymptotics for OTM VIX options.

Before addressing the short-maturity asymptotics of VIX options we discuss the forward VIX $F_V(T)$ defined as in (\ref{VIX:F}). In the $\tau \to 0$ limit this becomes $\mathbb{E}[v_T]$. Due to the explosion of $v_t$ in \eqref{vsde}, this expectation is infinite (Remark~\ref{remark:explosion}). However, the expectation of the capped process gives a finite but $a$-dependent value, which we denote $F_V(T,a) := \mathbb{E}[v_T]$
with the capped process $v_T$ given in \eqref{vsde:modification}.
We prove next a result for the short-maturity limit of $F_V(T,a)$ at finite $a$.

\begin{proposition}\label{prop:FVT:smallT}
We have
\begin{equation}
\lim_{T\to 0} F_V(T,a) = v_0\,.
\end{equation}
\end{proposition}

\begin{proof}
For any $T>0$, it follows from \eqref{vsde:modification} that
\begin{equation}
v_{T}=v_{0}e^{\int_{0}^{T}(\hat{\mu}_{V}(v_{s})-\frac{1}{2}\hat{\sigma}_{V}^{2}(v_{s}))ds+\int_{0}^{T}\hat{\sigma}_{V}(v_{s})dW_{s}},
\end{equation}
where $0\leq\hat{\sigma}_{V}(v)\leq a$
and $|\hat{\mu}_{V}(v)|\leq b$ for any $v$.
Therefore, we have
\begin{equation}
F_{V}(T,a)=\mathbb{E}[v_{T}]
\leq v_{0}e^{bT}\mathbb{E}\left[e^{-\int_{0}^{T}\frac{1}{2}\hat{\sigma}_{V}^{2}(v_{s})ds+\int_{0}^{T}\hat{\sigma}_{V}(v_{s})dW_{s}}\right]=v_{0}e^{bT},
\end{equation}
and
\begin{equation}
F_{V}(T,a)=\mathbb{E}[v_{T}]
\geq v_{0}e^{-bT}\mathbb{E}\left[e^{-\int_{0}^{T}\frac{1}{2}\hat{\sigma}_{V}^{2}(v_{s})ds+\int_{0}^{T}\hat{\sigma}_{V}(v_{s})dW_{s}}\right]=v_{0}e^{-bT},
\end{equation}
which completes the proof.
\end{proof}

This implies that for sufficiently small $T$, OTM VIX call options have $K>v_0$ and OTM VIX put options have $K< v_0$.

\begin{theorem}\label{thm:VIX:option}
Let $C_{V}(K,T)$ and $P_{V}(K,T)$ denote the VIX call and put option prices under the capped volatility model \eqref{vsde:modification}. 
Then
\begin{eqnarray}
&&\lim_{\tau,T\to 0} T \log C_V(K,T) = - J_V(K)\,,\quad K > v_0\,,\label{VIX:call:OTM} \\
&&\lim_{\tau,T\to 0} T \log P_V(K,T) = - J_V(K)\,,\quad K < v_0\,,\label{VIX:put:OTM}
\end{eqnarray}
where the rate function is
\begin{equation}\label{J:V:eqn}
J_{V}(K) = \frac12 \left( \int_{v_0}^K \frac{dz}{z \hat{\sigma}_V(z)} \right)^2\,.
\end{equation}    
\end{theorem}

\begin{proof}
By Proposition~\ref{prop:VIX:approx} and short-maturity asymptotics for European options in the local volatility model \cite{BBFpaper}, 
we can compute that for any $\tau>0$ and $K>v_{0}$,
\begin{align*}
\limsup_{T\rightarrow 0}T\log C_{V}(K,T)
&=\limsup_{T\rightarrow 0}T\log\mathbb{E}[(\mathrm{VIX}_{T}-K)^{+}]
\\
&\leq\limsup_{T\rightarrow 0}T\log\mathbb{E}[(v_{T}e^{b\tau+\frac{1}{2}a^{2}\tau}-K)^{+}]
\\
&=\limsup_{T\rightarrow 0}T\log\mathbb{E}[(v_{T}-Ke^{-b\tau-\frac{1}{2}a^{2}\tau})^{+}]
=-J_{V}\left(Ke^{-b\tau-\frac{1}{2}a^{2}\tau}\right),
\end{align*}
where $J_{V}(\cdot)$ is defined in \eqref{J:V:eqn}.
Similarly, we can compute that for any $\tau>0$ and $K>v_{0}$,
\begin{align*}
\liminf_{T\rightarrow 0}T\log C_{V}(K,T)
&=\liminf_{T\rightarrow 0}T\log\mathbb{E}[(\mathrm{VIX}_{T}-K)^{+}]
\\
&\geq\liminf_{T\rightarrow 0}T\log\mathbb{E}[(v_{T}e^{-b\tau}-K)^{+}]
\\
&=\liminf_{T\rightarrow 0}T\log\mathbb{E}[(v_{T}-Ke^{b\tau})^{+}]
=-J_{V}\left(Ke^{b\tau}\right).
\end{align*}
Since it holds for any $\tau>0$, by letting $\tau\rightarrow 0$, 
we showed that \eqref{VIX:call:OTM} holds.
Similarly, we can show that \eqref{VIX:put:OTM} holds. 
This completes the proof.
\end{proof}

In particular, for OTM call options, i.e. $K>v_{0}$, we have
\begin{equation}
J_{V}(K)=
\begin{cases}
\frac{1}{2}\left(\int_{v_{0}}^{K}\frac{dz}{z\sigma_{V}(z)}\right)^{2} &\text{if $\hat{v}\geq K>v_{0}$},
\\
\frac{1}{2}\left(\int_{v_{0}}^{\hat{v}}\frac{dz}{z\sigma_{V}(z)}+\frac{\log(K/\hat{v})}{a}\right)^{2} &\text{if $K>\hat{v}>v_{0}$},
\\
\frac{1}{2}\left(\frac{\log(K/v_{0})}{a}\right)^{2} &\text{if $K>v_{0}\geq\hat{v}$},
\end{cases}
\end{equation}
and for OTM put options, i.e. $K<v_{0}$, we have
\begin{equation}
J_{V}(K)=
\begin{cases}
\frac{1}{2}\left(\frac{\log(K/v_{0})}{a}\right)^{2} &\text{if $\hat{v}\leq K<v_{0}$},
\\
\frac{1}{2}\left(\int_{K}^{\hat{v}}\frac{dz}{z\sigma_{V}(z)}+\frac{\log(v_{0}/\hat{v})}{a}\right)^{2} &\text{if $K<\hat{v}<v_{0}$},
\\
\frac{1}{2}\left(\int_{K}^{v_{0}}\frac{dz}{z\sigma_{V}(z)}\right)^{2}  &\text{if $K<v_{0}\leq\hat{v}$}.
\end{cases}
\end{equation}

We define the implied volatility of the VIX options $\sigma_{\mathrm{VIX}}(K,T)$ in the capped volatility model \eqref{vsde:modification} as
\begin{align}\label{sigVIX:def}
    C_V(K,T) &= e^{-rT } c_{\mathrm{BS}}(K,T;F_V(T,a),\sigma_{\mathrm{VIX}}(K,T))\,, \\
    P_V(K,T) &= e^{-rT } p_{\mathrm{BS}}(K,T;F_V(T,a),\sigma_{\mathrm{VIX}}(K,T))\,, \nonumber
\end{align}
where $c_{\mathrm{BS}}(K,T;F,\sigma)$ and $p_{\mathrm{BS}}(K,T;F,\sigma)$ are the undiscounted Black-Scholes option prices with forward $F$ and volatility $\sigma$ and $F_V(T,a) = \mathbb{E}[v_T]$.

The short-maturity option pricing result of Theorem \ref{thm:VIX:option} translates in the usual way into a short-maturity result for the VIX implied volatility. 
Assume that the cap $a$ is sufficiently large, such that $\hat{v}>K$ and $\hat{v}>v_{0}$. Then, the short-maturity asymptotics for the VIX call and put options is equivalent with the
short-maturity asymptotics of the VIX implied volatility
\begin{equation}\label{BBF}
\lim_{T\to 0} \sigma_{\mathrm{VIX}}(x,T) := \sigma_{\mathrm{VIX}}(x)=\frac{\log(K/v_0)}{\int_{v_0}^K \frac{dz}{z \hat{\sigma}_V(z)}} = \frac{\log(K/v_0)}{\int_{v_0}^K \frac{dz}{z \sigma_V(z)}}\,,
\end{equation}
where $x = \log(K/v_0)$ is the log-strike. 
We used here the result of Proposition~\ref{prop:FVT:smallT} to take the limit $\lim_{T\to 0} F_V(T,a)=v_0$ for any finite $a$. 

The integral in the denominator in \eqref{BBF} can be evaluated exactly with the result
\begin{equation}
\int_{v_0}^K \frac{dz}{z \sigma_V(z)} = \frac{1}{\omega}
\left\{\arctanh\left(\frac{\rho(\beta-1) v_0 + \omega}{\sigma_V(v_0)}\right)
-\arctanh\left(\frac{\rho (\beta-1) K + \omega}{\sigma_V(K)}\right) 
\right\} \,.
\end{equation}

The VIX implied volatility \eqref{BBF} can be expanded in log-strike
\begin{equation}
\sigma_{\mathrm{VIX}}(x) = \sigma_{\mathrm{VIX}}(0) + s_{\mathrm{VIX}}\cdot x + \frac12 \kappa_{\mathrm{VIX}}\cdot x^2 +O(x^3)\,,
\end{equation}
as $x\rightarrow 0$, where the ATM level, skew and convexity of the VIX implied volatility are
\begin{align}\label{atmVIX}
&\sigma_{\mathrm{VIX}}(0) := \sigma_V(v_0)\,,\\
\label{skewVIX}
&s_{\mathrm{VIX}} := v_0 \frac{d}{dv} \sigma_{\mathrm{VIX}}(v_0) = v_0 (\beta-1) \frac{\rho\omega + (\beta-1) v_0}{2\sigma_{V}(v_0)} \,,
\end{align}
and 
\begin{equation}\label{cvxVIX}
\kappa_{\mathrm{VIX}} := \frac{v_0 (\beta-1)}{3\sigma_V^4(v_0)}
\left( 2\omega^3 \rho + (\beta-1) \omega^2 (4 + \rho^2) v_0 +
4 (\beta-1)^2 \omega \rho v_0^2 + (\beta-1)^3 v_0^3 \right) \,.
\end{equation}

These results for the ATM level, skew and convexity are reproduced by Proposition~6.2 in \cite{VIXpaper} by taking $\eta(x) = x^{\beta-1}$ in that result\footnote{Explicitly, the results \eqref{atmVIX}, \eqref{skewVIX} and \eqref{cvxVIX} are reproduced by substituting  
$\sigma=2\omega, V_0 = v_0^2, \eta_1 = \beta-1$ and $\eta_2=\frac12 (\beta-1)^2$ in Proposition 6.2 of \cite{VIXpaper}.}
although, strictly speaking, they do not follow from Proposition 6.2 of \cite{VIXpaper} since the CEV local volatility function $\eta(x)$ does not satisfy the technical conditions required for its validity.

\begin{remark}\label{rmk:vix}
For $\beta<1$ and $\rho<0$, the VIX skew (\ref{skewVIX}) is positive, which agrees with the empirical evidence: the observed VIX smile is up-sloping. 
On the other hand, the VIX convexity (\ref{cvxVIX}) is positive,
which unfortunately disagrees with empirical evidence: the observed VIX smile is concave. This disfavors the SABR model as a realistic model for VIX smiles.
\end{remark}
\vspace{0.2cm}

\textbf{Numerical example.} We illustrate the theoretical results with numerical simulations of VIX options under the capped volatility process \eqref{vsde:modification}. We assume the following parameters
\begin{equation}\label{params}
\omega = 1.0\,,\quad v_0 = 0.1\,,\quad
a = 2.0\,,\quad b = 1.0\,.
\end{equation}

\begin{table}
\caption{\label{tab:1} 
Numerical values for the simulation of the capped volatility process with parameters (\ref{params}). $\hat v$ is given by \eqref{vhat} and the last column shows the MC estimate for the forward VIX $F_V(T,a)$ at $T=0.1$.}
\begin{center}
\begin{tabular}{|c|cc|}
\hline
$\rho$ & $\hat v$ & $F_V(0.1,a)$ \\
\hline\hline
$-0.7$ & 2.336 & $0.1003 \pm 0.0001$ \\
0.0 & 3.464 & $0.1001 \pm 0.0001$ \\
0.7 & 5.136 & $0.0998 \pm 0.0001$ \\
\hline
\end{tabular}
\end{center}
\end{table}

In Table~\ref{tab:1} we show the values of the $\hat v$ parameter at which the cap on $\hat \sigma_V(v)$ is reached. For all correlation values, this is much larger than the spot volatility $v_0$. 

The SDE for the capped volatility process $v_t$ in \eqref{vsde:modification} was simulated numerically using an Euler scheme with $n=100$ time steps and $N_{\mathrm{MC}}=100k$ MC paths. Table \ref{tab:1} shows the VIX forward prices $F_V(T,a)$ for $T=0.1$ for several values of the correlation parameter $\rho$ obtained from the MC simulation. These values are close to $v_0=0.1$.

Using this simulation we priced VIX options with maturity $T=0.1$. The option prices were converted to VIX implied volatility using the definition \eqref{sigVIX:def}. The results for $\sigma_{\mathrm{VIX}}$ are shown in Figure~\ref{Fig:MC} as the red dots with error bars, for three values of the correlation $\rho \in \{-0.7,0,+0.7\}$. The solid black curve in these plots shows the short-maturity asymptotic VIX volatility (\ref{BBF}).

From these plots we note reasonably good agreement of the asymptotic result with the numerical simulation within the errors of the Monte Carlo simulation. These plots also illustrate the main features of the VIX smile under the SABR model noted above: for negative correlation $\rho<0$ the VIX smile is increasing, which is in agreement with empirical data. However, for $\rho<0$ the smile is convex, which is different from the observed concave shape of this smile. This suggests that the SABR model may not allow a precise calibration to VIX options market data, although it can be useful as a simple approximation. A detailed empirical study of the ability of the LSV with log-normal volatility (of which SABR is a limiting case) to calibrate to market data will be left as a future research direction.

\begin{figure}
\centering
\includegraphics[width=2.0in]{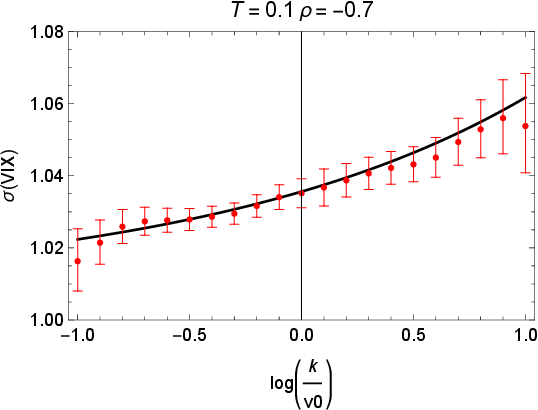}
\includegraphics[width=2.0in]{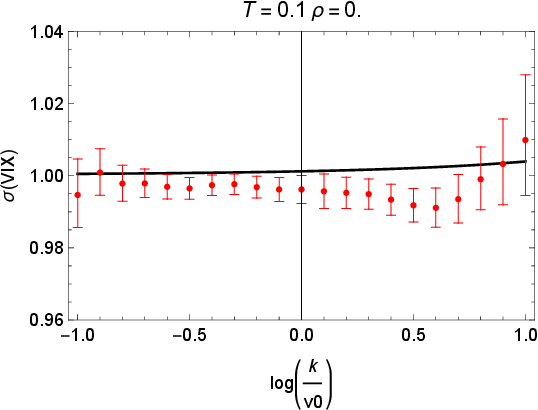}
\includegraphics[width=2.0in]{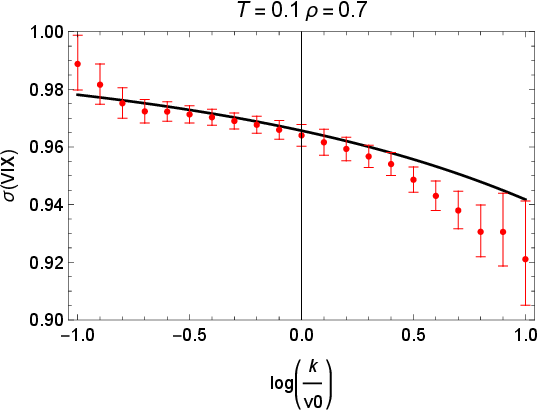}
\caption{The VIX implied volatility for the capped volatility model with parameters (\ref{params}) and three values of the correlation as shown. 
The red dots with error bars show the MC simulation and the black curve is the short maturity asymptotics (\ref{BBF}). The VIX options have maturity $T=0.1$.}
\label{Fig:MC}
\end{figure}

\section*{Acknowledgements}
We are grateful to an anonymous referee for helpful comments and suggestions. We would also like to thank Martin Forde and Alan Lewis 
for comments.
Lingjiong Zhu is partially supported by the grants NSF DMS-2053454, NSF DMS-2208303.

\appendix

\section{Proofs}

\begin{proof}[Proof of Proposition \ref{prop:p}]
Define the scale function 
\begin{equation}
p(x) = \int_0^x e^{-2F(y)} dy\,.
\end{equation}

i) The function $p$ is clearly increasing since $p'(x) = e^{-2F(x)} \geq 0$.
From (\ref{boundsF}), the scale function is bounded from above by the convergent integral
\begin{equation}
p(x) \leq \kappa \int_0^x 
\left( \frac{\omega^2}{R(y)} \right)^{\frac{\alpha}{(1-\beta)^2}} dy\,,
\end{equation}
where $\kappa$ is given in \eqref{kappa:defn}.
By the monotone convergence theorem, $p(x)$ converges to a finite limit $p_\infty$ since it 
is monotonically increasing and bounded from above.

ii) Large $x$ asymptotics. The upper and lower bounds on $p(x)$ are proportional to a common integral
\begin{equation}
I(x) := \int_0^x 
\left( \frac{\omega^2}{R(y)} \right)^{\frac{\alpha}{(1-\beta)^2}} dy\,.
\end{equation}
The large-$x$ asymptotics of this integral is
\begin{equation}
I(x) = c_0 - \frac{1}{1-\beta} \left( \frac{\omega}{1-\beta} \right)^{\frac{2-\beta}{1-\beta}}
\cdot \frac{1}{x^{\frac{1}{1-\beta}}}+o\left(x^{-1/(1-\beta)}\right)\,, \quad x\to \infty\,.
\end{equation}
The large $x$ asymptotics for $p(x)$ is the same, up to a multiplicative constant, which yields the result (\ref{largex}).

iii) Inverting the leading term asymptotics (\ref{largex}) of $p(x)$ for $x\to \infty$,  
gives the stated $y\to p_\infty$ asymptotics for $q(y)$.
\end{proof}

\begin{proof}[Proof of Proposition \ref{prop:explosion}]
Recall the Feller test for explosions of the solutions of a SDE. Define the function
\begin{equation}\label{nudef}
\nu(x) := \int_c^x p'(y) \int_c^y \frac{2dz}{p'(z) \sigma^2(z)} dy\,.
\end{equation}

\begin{theorem}[Feller’s (1952) Test of Explosion \cite{Feller1952}, Theorem~5.29 \cite{KS}]\label{thm:Feller}
Assume that the nondegeneracy \eqref{ND} and local integrability \eqref{LI} conditions
hold,  
and let $X_t$ be a weak solution in $I = (0,\infty )$ of the SDE \eqref{general:SDE} with nonrandom initial condition $X_0> 0$. Then $\mathbb{P}(\tau_{\infty}=\infty) = 1$ or
$\mathbb{P}(\tau_{\infty}=\infty) < 1$, according to whether $\nu(0+) = \nu(\infty -) = \infty$ or not, where $\tau_{\infty} = \sup\{
t\geq 0 : X_t < \infty\}$.
\end{theorem}

We prove an upper bound on the function $\nu(x)$, and will show that it is finite for all $x$. 
Using $p'(y) = e^{-2F(y)}$ and the bounds (\ref{boundsF}) on $e^{-2F}$ we have
\begin{equation}\label{nubound}
\nu(x) \leq \kappa \int_0^x \left( \frac{\omega^2}{R(y)}\right)^{\alpha/(1-\beta)^2}
\left( \int_0^y \left( \frac{\omega^2}{R(z)}\right)^{-\alpha/(1-\beta)^2}\frac{2dz}{z^2 R(z)} 
\right) dy\,.
\end{equation}

Recall from \eqref{R:eqn} that $R(y) = (1-\beta)^2 y^2 + 2\bar\omega (1-\beta) y + \omega^2$ 
and 
\begin{equation}
\frac{\alpha}{(1-\beta)^2} = \frac{2-\beta}{2(1-\beta)} \geq 0\,,\qquad\mbox{ for } 0 \leq \beta < 1\,.
\end{equation}

Denote the $z$ integral 
\begin{equation}
I_z(y) := \int_0^y
\left( \frac{\omega^2}{R(z)}\right)^{-\frac{\alpha}{(1-\beta)^2}} \frac{2dz}{z^2 R(z)} \,.
\end{equation}
The integrand has the large $z$ asymptotics $\sim \frac{2}{\omega^2} 
z^{-4 +\frac{2\alpha}{(1-\beta)^2} } = \frac{2}{\omega^2} 
z^{-\frac{2-3\beta}{1-\beta}} $ as $z\to \infty$. This gives the large $y$ asymptotics
\begin{equation}
I_z(y) = 
\begin{cases}
\mbox{const } + \frac{2(1-\beta)}{\omega^2(2\beta-1)} 
y^{\frac{2\beta-1}{1-\beta}}\,, & \beta\neq \frac12\,, \\
\mbox{const } + \frac{2}{\omega^2} \log y\,, & \beta = \frac12\,.
\end{cases}
\end{equation}

We distinguish between two cases:

i) $\beta=\frac12$. For this case the integral $I_z(y)$ has logarithmic growth as $y\to \infty$.
The bound (\ref{nubound})
has the large $x$ asymptotics, i.e. there exists some $\kappa'>0$ such that
\begin{equation}
\nu(x) \leq \kappa \int_0^x \left( \frac{\omega^2}{R(y)}\right)^{\alpha/(1-\beta)^2}I_z(y) dy 
\leq\kappa'
\int_{1}^x \frac{\log y}{y^{\frac{2-\beta}{1-\beta}}} dy \,,
\end{equation}
for any sufficiently large $x$, where $\kappa$ is given in \eqref{kappa:defn} and the integral is
\begin{equation}
\int_{1}^x \frac{\log y}{y^{\frac{2-\beta}{1-\beta}}} dy 
= \mbox{const }
- (1-\beta) \frac{\log x + (1-\beta)}{x^{\frac{1}{1-\beta}}}\,,
\end{equation}
which is bounded as $x\to \infty$.

We conclude that for this case $\nu(x)<\infty$ for all $x$.

ii) $\beta \neq \frac12$. For this case the integrand in the upper bound (\ref{nubound}) on 
$\nu(x)$ has the large $y$ asymptotic form
\begin{equation}
\frac{1}{y^{\frac{2-\beta}{1-\beta}}} \left( \mbox{ const } + \frac{1}{2\beta-1} \cdot
\frac{1}{y^\frac{1-2\beta}{1-\beta}} \right)\,,\quad y \to \infty\,.
\end{equation}
Consider the two cases of $\beta<\frac12 $ and $\beta> \frac12$ separately.

a) For $\beta<\frac12$ the second term in the brackets 
can be neglected relative to the constant. 
This gives $\nu(x) \leq \nu_0 - \nu_1 \frac{1}{x^{\frac{1}{1-\beta}}}$ as $x\to \infty$ with $\nu_{0},\nu_{1}>0$ which is thus bounded as $x\to \infty$.

b) For $\beta > \frac12$ the second term in the brackets dominates over the first term, 
and the 
integrand has the form $\frac{dy}{y^3}$. The large $x$ asymptotics of the integral giving the upper bound on $\nu(x)$ is
$\nu(x) \leq \nu_0 - \nu_1 \frac{1}{x^2}$, which has a finite limit as $x\to \infty$.

For both cases, the function $\nu(x)$ is bounded from above by a finite value,
which implies $\nu(x) <\infty$ for all $x$. By the Feller test of explosion (Theorem~\ref{thm:Feller}), this
implies that the process $v_t$ explodes in finite time with non-zero probability.
\end{proof}

\begin{proof}[Proof of Proposition \ref{prop:martingale}]
The proof uses the following result, due to Sin \cite{Sin}: consider the auxiliary volatility process
\begin{equation}\label{aux}
d\hat \sigma_t = b(\hat \sigma_t) dt + a(\hat \sigma_t) d\hat W_t\,,
\end{equation}
with
\begin{equation}
b(v) = \beta v^2 \Big[\frac12 (\beta-1) v + \rho \Big]\,,\quad
a(v) = v \sigma_V(v) \,.
\end{equation}
Then the asset price $S_t$ is a true martingale if and only if $\hat \sigma_t$ does not explode in finite time.

Define the scale function for the process \eqref{aux}
\begin{equation}\label{phatdef}
\hat p(x) = \int_0^x e^{-2\int_0^\xi \frac{b(y)}{a^2(y)} dy} d\xi\,.
\end{equation}
A sufficient condition for the absence of explosion of $\hat \sigma_t$ is $\hat p(\pm \infty) = \pm \infty$, see p.~382 in \cite{KS}.
We will prove in the following that this holds indeed. 

The integral in the exponent of \eqref{phatdef} is
\begin{equation}
\hat F(\xi) := -2 \int_0^\xi \frac{b(y)}{a^2(y)} dy = 2\beta \int_0^\xi 
\frac{\frac12 (1-\beta) y - \rho}{\sigma_V^2(y)} dy\,.
\end{equation}
As noted, for $|\rho | \neq 1$, the denominator $\sigma_V^2(y)$ does not vanish for all real $y$.  

Following the same approach as in Section~\ref{sec:2.2},
we have the lower and upper bounds on this integral
\begin{equation}\label{two:bounds}
\kappa_- \left( \frac{\sigma_V^2(\xi)}{\omega^2} \right)^{\frac{\beta}{2(1-\beta)}} < 
e^{\hat F(\xi)} <
\kappa_+ \left( \frac{\sigma_V^2(\xi)}{\omega^2} \right)^{\frac{\beta}{2(1-\beta)}}\,,
\end{equation}
with $\kappa_\pm$ finite constants.
Substituting \eqref{two:bounds} into \eqref{phatdef} we get the lower bound on the scale function
\begin{equation}
\hat p(x) \geq \kappa_+ \int_0^x \left( \frac{\sigma_V^2(\xi)}{\omega^2} \right)^{\frac{\beta}{2(1-\beta)}} d\xi 
\to +\infty\mbox{ as } x \to +\infty\,,
\end{equation}
and the lower bound
\begin{equation}
\hat p(x) \leq \kappa_- \int_0^x \left( \frac{\sigma_V^2(\xi)}{\omega^2} \right)^{\frac{\beta}{2(1-\beta)}} d\xi \to -\infty
\mbox{ as }
 x \to -\infty\,.
\end{equation}
Thus for any $\beta<1$, we have $\hat p(\pm \infty) = \pm \infty$
which proves the stated claim. 
\end{proof}
\bibliographystyle{plain}
\bibliography{VIXSABR}

\end{document}